\pdfoutput=1
\documentclass[orivec,envcountsame]{llncs}

\usepackage[T1]{fontenc} %
\usepackage[protrusion=true,expansion=true]{microtype}

\usepackage{ownstyle}

\title{Initiality for Typed Syntax and Semantics}
\author{Benedikt Ahrens}
\institute{Universit\'e Nice Sophia Antipolis, France}

\newcommand{\sourceurl}{\url{http://math.unice.fr/laboratoire/logiciels}}

\begin{document}

\maketitle

\begin{abstract}
 We give an algebraic characterization of the syntax and semantics 
of a class of simply--typed languages, such as the language $\PCF$: we characterize simply--typed binding syntax equipped with reduction rules
via a universal property, namely as the initial object of some category.
For this purpose, we employ techniques developed in two previous works: in \cite{ahrens_ext_init}, we 
model syntactic translations between languages \emph{over different sets of types} as initial morphisms
in a category of models. In \cite{ahrens_relmonads}, we characterize untyped syntax \emph{with reduction rules}
as initial object in a category of models.
In the present work, we show that those techniques are modular %
enough to be combined:
we thus characterize
simply--typed syntax with reduction rules as initial object in a category. 
The universal property yields an operator which allows to specify
translations --- that are semantically faithful by construction --- 
between languages over possibly different sets of types. %

We specify a language by a \emph{2--signature}, that is, a signature on two levels: 
the \emph{syntactic} level specifies the types and terms of the language, and 
associates a type to each term.
The \emph{semantic} level specifies, through \emph{inequations}, reduction rules
on the terms of the language.
To any given 2--signature we associate a category of models.
We prove that this category has an initial object, which integrates the
types and terms freely generated by the 2--signature, and the reduction relation on those terms
generated by the given inequations.
We call this object the \emph{(programming) language generated by the 2--signature}.
\end{abstract}


\section{Introduction}

We give a characterization, via a universal property, of the syntax and semantics of 
simply--typed languages with variable binding.
More precisely, we characterize the terms and sorts associated to a signature 
equipped with reduction rules as the initial object in a category of models.
Initiality in this category gives rise to an iteration principle (cf.\ \autoref{rem:comp_sem_iteration})
which allows to specify translations between languages in a convenient way as initial morphisms.
The category of models is sufficiently large --- and thus the iteration principle stemming from 
initiality is sufficiently general --- to account for translations between 
languages over different sets of sorts.
Furthermore, translations specified via this principle are ensured to be faithful with respect to 
reduction in the source and target languages, as well as compatible in a suitable sense
with substitution on either side.

To illustrate the iteration operator stemming from initiality, we use it to 
specify a translation from $\PCF$ to the untyped lambda calculus $\ULC$. 
We do so in the proof assistant \textsf{Coq} \cite{coq}; for this purpose, we prove formally, in \textsf{Coq}, an instance of our main theorem
for the 2--signature of $\PCF$: 
the types and terms of $\PCF$, equipped with their usual reductions,
form an initial object in the category of models of $\PCF$. We then use the iteration principle to 
obtain an initial morphism --- a translation, faithful with respect to reductions --- to $\ULC$, as an executable \textsf{Coq} function.
The \textsf{Coq} theory files as well as documentation are available online at
 \sourceurl.

\subsubsection{Summary}

We define a notion of \emph{2--signature} which allows the specification of the \emph{types and terms}
of a language --- via an underlying 1--signature --- as well as its \emph{semantics} in form of reduction rules.
A 1--signature $(S,\Sigma)$ is given by a pair of a signature $S$ for types and a binding signature $\Sigma$ for 
terms typed over the set of types associated to $S$.
Reduction rules for terms generated by $\Sigma$ are specified via a set $A$ of inequations over $(S,\Sigma)$.
A 2--signature $((S,\Sigma),A)$ is a pair of a 1--signature $(S,\Sigma)$ and a set $A$ of inequations over $(S,\Sigma)$.
To such a 2--signature we associate a category of representations, for which the types and terms generated by $(S,\Sigma)$, 
equipped with reductions according to $A$, forms an initial object.

\emph{1--signatures} are defined in \cite{ahrens_ext_init}. There, we associate a category $\Rep(S,\Sigma)$  
of representations to any 1--signature $(S,\Sigma)$, and show that the types and terms
freely generated by $(S,\Sigma)$ form an initial object in this category.
Representations there are built from monads on families of sets.
In the present work, we build a different category $\Rep^{\Delta}(S,\Sigma)$ of representations using 
\emph{relative} monads from sets to preordered sets, which allows --- in a second step, cf.\ below --- the integration 
of reduction rules to account for semantic aspects.
The two categories of representations, $\Rep(S,\Sigma)$ and $\Rep^{\Delta}(S,\Sigma)$, are connected through
an adjunction which transports the initial object of the former to the latter category (cf.\ \autoref{lem:init_no_eqs_typed}).

\emph{Inequations} over untyped 1--signatures are considered in \cite{ahrens_relmonads}. There, we define a notion of
2--signature for untyped syntax with semantics in form of reduction rules and show that its associated category of representations 
has an initial object.
In the present work, we define inequations over \emph{typed} 1--signatures as defined in \cite{ahrens_ext_init}.
Given a set $A$ of inequations over a 1--signature $(S,\Sigma)$, 
the representations of $(S,\Sigma)$ that \emph{satisfy} each inequation of $A$, form a full subcategory of $\Rep^{\Delta}(S,\Sigma)$,
which we call the
category of representations of $(S,\Sigma,A)$.
Our main theorem (cf.\ \autoref{thm:init_w_ineq_typed}) 
states that this category has an initial object, which integrates the types and terms freely generated by $(S,\Sigma)$,
equipped with reduction rules generated by the inequations of $A$.

\subsubsection{Related Work}\label{sec:rel_work}

Related work is reviewed extensively in \cite{ahrens_relmonads,ahrens_ext_init}, as well as in the author's 
PhD thesis \cite{ahrens_phd}.
We give a brief overview:
rewriting in nominal settings is examined by F\'ernandez and Gabbay \cite{fernandez_gabbay_nominal_rewriting}.
Ghani and L\"uth \cite{DBLP:journals/njc/GhaniL03} present rewriting for algebraic theories without variable binding;
they characterize 
equational theories 
resp.\ rewrite systems 
 as \emph{coequalizers} resp.\ \emph{coinserters} in a category of monads on
the categories $\Set$ resp.\ $\PO$. 
Fiore and Hur \cite{DBLP:conf/icalp/FioreH07} have extended Fiore's work to integrate semantic aspects into initiality results.
In particular, Hur's thesis \cite{hur_phd} is dedicated to \emph{equational} systems for syntax with variable binding.
In a ``Further research'' section \cite[Chap.\ 9.3]{hur_phd}, Hur suggests the use of preorders, or more generally, 
arbitrary relations to model \emph{in}equational systems.
Hirschowitz and Maggesi \cite{DBLP:conf/wollic/HirschowitzM07} prove initiality of the set of lambda terms modulo beta and eta conversion
in a category of \emph{exponential monads}.
In an unpublished paper \cite{journals/corr/abs-0704-2900}, they define a notion of \emph{half--equation} and 
\emph{equation} to express congruence between terms. We adopt their definition in this paper, but interpret a pair
of half--equations as \emph{in}equation rather than equation.

\section{Relative Monads and Modules}

\emph{Relative monads} were defined by Altenkirch et al.\ \cite{DBLP:conf/fossacs/AltenkirchCU10} to overcome 
the restriction of (regular) monads to \emph{endo}functors.
In an earlier work \cite{ahrens_relmonads}, we define morphisms of relative monads and \emph{modules over relative monads}.
In the following section we define a more general notion of \emph{colax} morphism of relative monads ---
which we use in \autoref{chap:comp_types_sem} to model translations between languages over different sets of types ---
and generalize constructions of \cite{ahrens_relmonads} to such colax morphisms. 
Some definitions from \cite{ahrens_relmonads,ahrens_ext_init} which we use in the present work, 
are recalled at the beginning.

We denote by $\Set$ the category of sets and total maps of sets. We call $\PO$ the 
category of preordered sets and monotone maps between them. 

\begin{definition}
\label{def:delta}
\label{lem:adj_set_po}
 We call $\Delta : \SET\to\PO$ the left adjoint of the forgetful functor $U : \PO\to\Set$. 
 The functor $\Delta$ associates  
 to each set $X$ the set itself together with the smallest preorder, i.e.\ the diagonal of $X$,
 $\Delta(X) := (X,\delta_X)$. 


\end{definition}

\begin{definition}[Category of Families]\label{def:TST}\label{def:TS}\label{rem:adj_set_po_typed}
 Let $\C$ be a category and $T$ be a set, i.e. a discrete category.
 We denote by $\family{\C}{T}$ the functor category, an object of which is a $T$--indexed family of objects of $\C$.  
We write $V_t := V(t)$ for objects and morphisms.
 Given a functor $F : \C\to \D$, we denote by $\family{F}{T} : \family{\C}{T} \to \family{\D}{T}$ the induced functor.
%
%

\end{definition}

\begin{definition}[Relative Monad on $\TDelta{T}$, enriched]
\label{def:strenghtened_delta}
 We strengthen the definition of a relative monad $P$ on $\TDelta{T}$ by requiring the 
 substitution map $\sigma_{X,Y}$ to be monotone with respect to the preorders induced by the preorders on $PY$,
 \[ \sigma_{X,Y} : \TP{T}(\Delta X,PY) \to \TP{T}(PX,PY) \enspace . \]
\end{definition}

\noindent
From now on, a relative monad on $\TDelta{T}$ is meant to be enriched in 
the sense of \autoref{def:strenghtened_delta}, i.e.\ monotone in both the first-- and the higher--order
argument.

\begin{example}[Lambda Calculus as Relative Monad on $\TDelta{T}$]\label{ex:ulcbeta}
Let $T := \TLCTYPE$ be the set of types of the simply--typed lambda calculus, built from a base type and a binary arrow constructor.
Given a set family $V\in \TS{\TLCTYPE}$, we denote by $\TLC(V) \in \TS{\TLCTYPE}$ the set family of simply--typed lambda terms over
$\TLCTYPE$ in context $V$, which might be implemented in the proof assistant \textsf{Coq} as follows: 
\begin{lstlisting}
Inductive TLC (V : T -> Type) : T -> Type :=
  | Var : forall t, V t -> TLC V t
  | Abs : forall s t TLC (V + s) t -> TLC V (s ~> t)
  | App : forall s t, TLC V (s ~> t) -> TLC V s -> TLC V t.
\end{lstlisting}
Here \lstinline!V + s! is a notation denoting the context \lstinline!V! extended by a fresh variable of type \lstinline!s! --- the variable 
that is bound by the constructor \lstinline!Abs s t!.
 We occasionally leave the object type arguments of the constructors implicit and
 write $\lambda M$ and $M(N)$ for \lstinline!Abs s t M! and \lstinline!App s t M N!, respectively.
The set family of lambda terms is equipped with a structure of a monad 
$\TLC$ on the category $\TS{\TLCTYPE}$ as follows \cite{DBLP:journals/iandc/HirschowitzM10}:
the family $\we^{\TLC}$ is given by the family of constructors \lstinline!Var!, 
and the substitution map is given by capture--avoiding 
simultaneous substitution:
\[\sigma_{X,Y} : \TS{T}\bigl(X, \TLC(Y)\bigr) \to \TS{T}\bigl(\TLC(X),\TLC(Y)\bigr) \enspace .\]

\noindent
Similarly, with the same operations $\eta$ and $\sigma$, 
we can consider it as a \emph{relative} monad on the functor $\TDelta{\TLCTYPE}$,
\[ \TLCDELTA : \TS{\TLCTYPE} \to \TP{\TLCTYPE} \enspace . \]
 The underlying object map $\TLCDELTA$ associates, to each 
set family $V$, the family of lambda terms in context $V$, equipped with the diagonal preorder,
corresponding to \emph{syntactic equality}:
 \[\TLCDELTA : V \mapsto \left(\TLC(V), \delta_{\TLC(V)}\right) \enspace . \]

\noindent
We equip each set $\TLC(V)(t)$ of lambda terms over context $V$ of object type $t$  with a preorder 
taken as the reflexive--transitive closure of the 
relation generated by the \emph{beta} rule
\begin{equation} \lambda  M (N) ~ \leq ~ M [*:= N] \tag{($\beta$)} \label{eq:beta_rule}\end{equation}
and its propagation into subterms:
 \[ \TLCB : V \mapsto \left(\TLC(V), \beta^{*}_{\TLC(V)}\right) \enspace . \]
The beta rule in \autoref{eq:beta_rule} states that the application of a lambda abstraction with body $M$ to an argument $N$
reduces to the term $M$ in which the term $N$ is substituted for the fresh variable of $M$ --- recall from above that
 $M$ lives in an extended context ---  in a capture--avoiding manner.
 This assignment defines a relative monad $\TLCB$ on the functor $\TDelta{T}:\TS{T}\to \TP{T}$.

\end{example}

\emph{Modules over relative monads} and their morphisms are defined in \cite{ahrens_relmonads}, together with several constructions of modules. 
Recall that modules over $P$ with codomain $\E$ and morphisms between them form a category called $\RMod{P}{\E}$.
We give some examples of modules and module morphisms over the monad  $\TLCB$,
which hold analogously for the monad $\TLCDELTA$:

\begin{example}[\autoref{ex:ulcbeta} cont.]
\label{ex:ulcb_taut_mod}\label{ex:ulcb_der_mod}\label{ex:ulcb_prod_mod}\label{ex:ulcb_constructor_mod_mor}
 The map $\TLCB : V \mapsto \TLCB(V)$ yields a module over the relative monad $\TLCB$, the \emph{tautological $\TLCB$--module} $\TLCB$.
 Given $V \in \TS{T}$ and $s\in T$, we denote by $V ^ s$ the context $V$ enriched by an additional variable of type $s$.
The map $\TLCB^s : V \mapsto \TLCB(V^s)$ inherits the structure of a $\TLCB$--module from the
  tautological module $\TLCB$. 
We call $\TLCB^s$ the \emph{derived module with respect to $s\in T$} of the module $\TLCB$. 
 Given $t \in T$, the map $V\mapsto \TLCB(V)(t) : \TS{T} \to \PO$ inherits a structure of a $\TLCB$--module, the 
\emph{fibre module $\fibre{\TLCB}{t}$ with respect to $t\in T$}.
 Given $s,t\in T$, the map $V\mapsto \TLCB(V)(s \TLCar t)\times \TLCB(V)(s)$ inherits a structure of a $\TLCB$--module.
 Finally, the constructors of abstraction \lstinline!Abs s t! and application \lstinline!App s t! 
are carriers of morphisms of $\TLCB$--modules:
 \begin{align*} \Abs_{s,t} : \fibre{\TLCB^{s}}{t} \to \fibre{\TLCB}{s\TLCar t} \enspace , \quad
                \App_{s,t} : \fibre{\TLCB}{s\TLCar t} \times \fibre{\TLCB}{s} \to \fibre{\TLCB}{t} \enspace .
 \end{align*}
Analogous remarks hold for the monad $\TLCDELTA$ and modules over this monad.
\end{example}

\noindent
As in \autoref{ex:ulcbeta}, we consider a language over a set $T$ of types as a (relative) monad on $\TDelta{T}$.
Translations between languages are given by \emph{colax morphisms of monads}:

\begin{definition}
 \label{def:colax_rel_mon_mor}
Suppose given two relative monads $P : \C\stackrel{F}{\to}\D$ and 
$Q : \C'\stackrel{F'}{\to}\D'$. A \emph{colax morphism of relative monads} from $P$ to $Q$ is a
quadruple $h = (G,G',N,\tau)$ of
a functor $G\colon \C\to\C'$, a functor $G': \D \to \D'$ as well as a natural transformation
$N: F'G \to G'F$ and a natural transformation $\tau : PG' \to GQ$
such that the following diagrams commute for any objects $c,d$ and any suitable morphism $f$:
\begin{equation*} 
 \begin{xy}
  \xymatrix @C=5pc @R=2pc {
  G'Pc \ar[r]^{G'\kl[P]{f}} \ar[d]_{\tau_c}& G'Pd \ar[d]^{\tau_d} \\
  QGc \ar[r]_{\kl[Q]{\comp{Nc}{\comp{G'f}{\tau_d}}}} & QGd 
}
 \end{xy}
\qquad
\begin{xy}
  \xymatrix @=2pc {
 F'Gc \ar[r]^{Nc} \ar[rrd]_{\we^Q_{Gc}} & 
                 G'F c \ar[r]^{G'\we^P_c}& G'Pc \ar[d]^{\tau_c} \\
{} & {} & QGc.
}
\end{xy}
\end{equation*}
\end{definition}

Given a morphism of relative monads $h : P \to Q$ and a $Q$--module $N$ with codomain $\E$, 
we define the \emph{pullback} $P$--module $h^*N$, also with codomain $\E$:
\begin{definition}
We define the \emph{pullback of $N$ along $h$} 
with object map $c\mapsto M (Gc)$
and with substitution map, for $f : Fc \to Pd$, as
 $\mkl[h^*M] f := \mkl[M]{\comp{N_c}{\comp {G'f}{\tau_d}}}$.
The pullback extends to module morphisms and is functorial. 
\end{definition}

Given two languages over different object types $T$ and $T'$, modelled as relative monads $P$ and $Q$
on $\TDelta{T}$ and $\TDelta{T'}$, respectively, we model a translation from $P$ to $Q$
by a colax monad morphism whose underlying functors are \emph{retyping functors}:

\begin{definition}[Retyping Functor]\label{rem:retyping_adjunction_kan}
Let $g:T\to T'$ be a map of sets, and let $\C$ be a cocomplete category.
The map $g$ induces a functor $g^*:\family{\C}{T'} \to \family{\C}{T}$ by postcomposition,
 $W \mapsto \comp{g}{W}$.
  The \emph{retyping functor $\retyping{g}:\family{\C}{T} \to \family{\C}{T'}$ associated to $g:T\to T'$}
  is defined as the left Kan extension operation 
  along $g$, that is, we have an adjunction $\retyping{g} \dashv g^*$.
\end{definition}

\begin{remark} \label{rem:rel_mon_mor_case}
  We are going to use the following instance of \autoref{def:colax_rel_mon_mor}:
  $P$ and $Q$ are monads --- e.g., languages --- on $\TDelta{T}$ and $\TDelta{T'}$, for sets $T$ and $T'$ of object types.
  The functors $G$ and $G'$ are the retyping functors (cf.\ \autoref{rem:retyping_adjunction_kan}) associated to some 
  translation of types $g : T \to T'$, and $N$ is the identity transformation. Then $\tau$ denotes a translation of terms
  from $P$ to $Q$:
 \[
 \begin{xy}
  \xymatrix @C=2pc @R=1.5pc {
       **[l] \TS{T} \ar[r]^{\TDelta{T}} \ar[d]_{\retyping{g}}   & **[r] \TP{T} \ar[d]^{\retyping{g}}\\
       **[l] \TS{T'}\ar[r]_{\TDelta{T'}} & **[r]\TP{T'} \ultwocell<\omit>{\Id}
}
 \end{xy} \qquad
\begin{xy}
  \xymatrix @C=2pc @R=1.5pc {
       **[l] \TS{T} \ar[r]^{P} \ar[d]_{\retyping{g}} \drtwocell<\omit>{\tau}  & **[r] \TP{T} \ar[d]^{\retyping{g}}\\
       **[l] \TS{T'}\ar[r]_{Q} & **[r]\TP{T'}. 
}
 \end{xy}
\]
\end{remark}

\noindent
A family of constructors, such as the family $(\Abs_{s,t})_{s,t\in \TLCTYPE}$ of \autoref{ex:ulcb_constructor_mod_mor},
is modelled via a family of module morphisms of suitable domain and codomain.
Equivalently, via \emph{uncurrying},
we can consider such a family as \emph{one} module morphism between two suitable modules:
intuitively, the idea is to write $\Abs(V,s,t) : \TLCB(V^s)(t) \to \TLCB(V)({s\TLCar t})$ 
instead of $\Abs_{s,t}(V)$.
For this to work, an object of the domain category of the source and target modules of $\Abs$ must be
of the form $(V,s,t)$, where $V$ is a context and $s,t\in \TLCTYPE$. More generally:

\begin{definition}[Pointed index sets]\label{def:cat_indexed_pointed}\label{def:cat_set_pointed}
\label{def:retyping_functor_pointed}
  Given a category $\C$, a set $T$ and a natural number $n$, we denote by $\family{\C}{T}_n$ the category
  with, as objects, diagrams of the form
    $n \stackrel{\vectorletter{t}}{\to} T \stackrel{V}{\to} \C$,
  written $(V, t_1, \ldots, t_n)$ with $t_i := \vectorletter{t}(i)$.
  A morphism $h$ to another such $(W,\vectorletter{t}) $
  with the same pointing map $\vectorletter{t}$ is given by a morphism $h : V\to W$ in $\family{\C}{T}$.
  
Any functor $F : \family{\C}{T} \to \family{\D}{T}$ extends to $F_n : \family{\C}{T}_n \to \family{\D}{T}_n$ via
   $ F_n (V,\vectorletter{t}) := (FV, \vectorletter{t})$.
Also, any relative monad $R$ over $F$ induces a monad $R_n$ over $F_n$.

Given a map of sets $g:T\to T'$, by postcomposing the pointing map with $g$, 
the retyping functor (cf.\ \autoref{rem:retyping_adjunction_kan}) generalizes to the functor
 \[ \retyping{g}(n) : \family{\C}{T}_n \to \family{\C}{T'}_n \enspace , 
\quad (V, \vectorletter{t}) \mapsto \left(\retyping{g} V, \comp{\vectorletter{t}}{g}\right) \enspace .  \] 
\end{definition}

\noindent
 The category $\family{\C}{T}_n$ consists of $T^n$ copies of $\family{\C}{T}$, which do not interact.
 Due to the ``markers'' $(t_1, \ldots, t_n)$ we can act differently on each copy, 
  cf. \autorefs{def:derived_rel_mod_II} and \ref{def:fibre_rel_mod_II}.

Two important constructions on modules over monads of \cite{ahrens_ext_init},
derivation and fibre modules, carry over to 
modules over monads on $\family{\Delta}{T}$.
Intuitively, derivation corresponds to considering terms in an extended context, whereas the fibre corresponds to picking
terms of a specific object type. Since we consider varying sets of types, the object type for context extension and fibre
is chosen through a natural transformation, which picks an element of any set.

Given $u\in T$, we denote by $D(u)\in \TS{T}$ the context with $D(u)(u) = \{*\}$ and $D(u)(t) = \emptyset$ for $u \neq t$.
For a context $V \in \TS{T}$ we set $V^{*u} := V + D(u)$.

Given a category $\C$ and $n\in \mathbb{N}$, we denote by $\T \C_n$ the category an object of which is a triple $(T,V,\vectorletter{t})$ of a set $T$,
a $T$--indexed family $V$ of objects of $\C$ and a vector $\vectorletter{t}$ of length $n$ of elements of $T$.
Note that for a fixed set $T$, the category $\family{\C}{T}_n$ is the fibre over $T$ of the forgetful functor $\T U_n : \T\C_n \to \Set$ which 
maps an object $(T,V,\vectorletter{t})$ to its indexing set $T$.
  Let $1 : \T\C_n \to \Set$ be the constant functor mapping to the singleton set.
 For a natural transformation $\tau : 1 \Rightarrow \T U_n$, 
we write $\tau (T,V,\vectorletter{t}) := \tau (T,V,\vectorletter{t})(*) \in T$, i.e.\ we omit 
the argument from the singleton set. Intuitively, such $\tau$ picks an element of $T$ of any object $(T,V,\vectorletter{t})\in \T \C_n$.

\begin{example}\label{ex:numbers_transformation}
 For $ 1 \leq k \leq n$, we denote by $k : 1 \Rightarrow \T U_n : \T\C_n \to \Set$ the natural transformation such that $k(T,V,\vectorletter{t})(*) := \vectorletter{t}(k)$.
\end{example}

\begin{definition}[Context Extension]
\label{def:derived_rel_mod_II}
Let $\tau$ be as above, and let $T$ be a fixed set.
Given a monad $P$ on $\TDelta{T}_n$ and a $P$--module $M$ with codomain $\E$, 
we define the \emph{derived module of $M$ with respect to $\tau$} by setting $M^\tau(V,\vectorletter{t}) := M (V^{*\tau(T,V,\vectorletter{t})}, \vectorletter{t})$.
\end{definition}

\begin{definition}[Fibre]
\label{def:fibre_rel_mod_II}
 Let $\tau$ be as in \autoref{def:derived_rel_mod_II}. Let $P$ be a monad on $\TDelta{T}_n$ and $M$ be a $P$--module
 with codomain category $\family{\E}{T}_n$.
 The \emph{fibre module $\fibre{M}{\tau}$ of $M$ with respect to $\tau$} has object map
  $ (V,\vectorletter{t}) \mapsto M(V,\vectorletter{t})(\tau(T,V,\vectorletter{t}))$, that is,
 the component $\tau(T,V,\vectorletter{t})$ of $M$, forgetting also the pointing map $\vectorletter{t}$.
\end{definition}

\begin{example}[\autoref{ex:ulcb_constructor_mod_mor} cont.]\label{ex:tlc_mod_higher_degree}
 Let $T := \TLCTYPE$.
 According to \autoref{def:cat_indexed_pointed}, we have a relative monad --- and its associated 
tautological module --- ${\TLCB}_2$ on the functor 
$\TDelta{T}_2 : \TS{T}_2 \to \TP{T}_2$.
Let $i : 1 \Rightarrow TU_2 : T\C_2 \to \Set$, for $i = 1, 2$, be a natural transformation as in \autoref{ex:numbers_transformation}.
Then we have a ${\TLCB}_2$--module 
\[{{\TLCB}_2}^1 : (V,s,t) \mapsto \TLCB_2(V^s,s,t) \enspace . \]
We also have a ${\TLCB}_2$--module 
\[ \fibre{{\TLCB}_2}{2} : (V,s,t)\mapsto \TLCB_2(V,s,t)(t) \enspace . \]
Again, as in \autoref{ex:ulcb_constructor_mod_mor}, analogous remarks hold for $\TLCDELTA$.
\end{example}

\begin{remark}[Module of Higher Degree corresponds to a Family of Modules]
  \label{rem:family_of_mods_cong_pointed_mod_relative}
   Let $T$ be a set and let $R$ be a monad on the functor $\family{\Delta}{T}$.
    Then a module $M$ over the monad $R_n$ corresponds precisely to a family of $R$--modules 
    $(M_{\vectorletter{t}})_{\vectorletter{t}\in T^n}$ by (un)currying.
  Similarly, a morphism $\alpha:M\to N$ of modules of degree $n$ is equivalent to a family 
   $(\alpha_{\vectorletter{t}})_{\vectorletter{t}\in T^n}$ of morphisms of modules of degree zero with
     $\alpha_{\vectorletter{t}}:M_{\vectorletter{t}}\to N_{\vectorletter{t}}$.
\end{remark}

\section{Signatures, Representations, Initiality}\label{chap:comp_types_sem}

We combine the techniques of \cite{ahrens_ext_init} and \cite{ahrens_relmonads} in order
to obtain an initiality result for simple type systems with reductions on the term level.
As an example, we specify, via the iteration principle stemming from the universal property, a 
semantically faithful translation 
from $\PCF$ with its usual reduction relation to the untyped lambda calculus with beta reduction.
Analogously to \cite{ahrens_relmonads}, we define a notion of \emph{2--signature} with two levels:
a \emph{syntactic} level specifying types and terms of a language, and, on top, a \emph{semantic} level
specifying reduction rules on the terms. 

The syntactic level itself --- given by a \emph{1--signature} $(S,\Sigma)$, 
cf.\ \autoref{def:1--sig} --- specifies the \emph{types} of the language, 
via an algebraic signature $S$,
as well as terms that are typed over the types specified by $S$, via a signature $\Sigma$ over $S$.
In a first result (cf.\ \autoref{lem:init_no_eqs_typed}) we characterize the language
generated by a 1--signature, and equipped with the \emph{equality preorder}, as an initial object
of a category of representations. An instance of this theorem is given in \autoref{ex:tlc_mod_mor_higher_degree},
where $\TLCDELTA$, equipped with two module morphisms given by the constructors \lstinline!Abs! and \lstinline!App!,
is characterized as the initial representation of a suitable 1--signature.

Afterwards we equip 1--signatures with \emph{inequations}, yielding 2--signatures (cf.\ \autoref{def:2--sig}). 
We prove an initiality result for 
those 2--signatures (cf.\ \autoref{thm:init_w_ineq_typed}), 
an instance of which characterizes the simply--typed lambda calculus $\TLCB$ 
with beta reduction as initial representation (cf.\ \autoref{ex:rep_2--sig_tlc}).

\subsubsection{Signatures for Types}

We consider sets of types that are specified by algebraic signatures, which 
are presented in \cite{ahrens_ext_init}. We review briefly:
\begin{definition}[Algebraic Signature]
 An \emph{algebraic signature} is given by a family of natural numbers.
\end{definition}
Intuitively, each natural number of such a family specifies the number of arguments of its associated type constructor.

\begin{example}\label{ex:type_sig_SLC}
 The types of the simply--typed lambda calculus are specified via the 
 algebraic signature $S_{\SLC} := \{ * : 0,\enspace (\TLCar) : 2\}$.
\end{example}

\begin{example}\label{ex:type_PCF}
  The language $\PCF$ \cite{Plotkin1977223} 
is a simply--typed lambda calculus with a fixed point operator
  and arithmetic constants.
The signature of the types of \PCF~is given by  
   $S_{\PCF}:= \lbrace\Nat:0 ,\enspace  \Bool: 0 ,\enspace (\PCFar): 2 \rbrace$.
 A representation $T$ of $S_{\PCF}$ is given by a set $T$ and three operations of suitable arities. 
 A morphism of representations is a map of sets compatible with the operations on either side.
\end{example}

\begin{lemma}\label{lem:initial_sort}
 Let $S$ be an algebraic signature. The category of representations of $S$ has an initial object $\hat{S}$. 
\end{lemma}

\subsubsection{Signatures for Terms}\label{sec:term_sigs}
For the rest of the section, let $S$ be a signature for types. 
Signatures for \emph{terms over $S$} are \emph{syntactically} defined as in \cite{ahrens_ext_init}.
We call \emph{degree} of an arity the number of object type variables appearing in the arity.
For instance, the signature $\Sigma_{\TLC}$ of simply--typed lambda terms over the signature $S_{\TLC}$ (cf.\ \autoref{ex:type_sig_SLC}) is given by
two arities of degree 2:
\begin{equation} \Sigma_{\TLC} := \{ \abs :  \bigl[([1],2)\bigr] \to (1\TLCar 2) \enspace , 
       \quad \app : \bigl[([],1\TLCar 2),([],1)\bigr]\to 2\}\enspace . 
  \label{eq:sig_tlc_higher_order}
\end{equation}
Intuitively, the numbers vary over object types. More precisely, for any representation of $S_{\TLC}$ in a set $T$,
the numbers vary over elements of $T$.

In order to define \emph{representations} of such a signature $(S,\Sigma)$, we need to 
consider set families where the indexing set is equipped with a representation of the type
signature $S$:

\begin{definition}
  Given a category $\C$ --- e.g., the category $\Set$ of sets --- 
  we define the category $S\C_n$ to be the category an object of which 
  is a triple $(T,V,\vectorletter{t})$ where $T$ is a representation of $S$,
  the object $V \in \family{\C}{T}$ is a $T$--indexed family of objects of $\C$ and $\vectorletter{t}$ is 
  a vector of elements of $T$ of length $n$.
We denote by $SU_n:S\C_n \to \Set$ the functor mapping an object $(T,V,\vectorletter{t})$ 
to the underlying set $T$.

  We have a forgetful functor $S\C_n \to \T\C_n$ which forgets the representation
 structure.
  On the other hand, any representation $T$ of $S$ in a set $T$ gives rise to a functor 
  $\family{\C}{T}_n \to S\C_n$, which ``attaches'' the representation structure.
\end{definition}

Recall from \cite{ahrens_ext_init} that $S(n)$ denotes terms of $S$ with free variables in $\{1,\ldots, n\}$.
The meaning of a term $s\in S(n)$ as a natural transformation $s: 1 \Rightarrow SU_n : S\C_n \to \Set$
is given by recursion on the structure of $s$:

\begin{definition}[Canonical Natural Transformation]\label{def:nat_trans_type_indicator}\label{def:canonical_nat_trans}
  Let $s\in S(n)$ be a type of degree $n$. 
  Then $s$ denotes a natural transformation 
       $s:1\Rightarrow SU_n : S\C_n \to \Set$
  defined recursively on the structure of $s$ as follows: for $s = \alpha (a_1,\ldots,a_k)$
  the image of a constructor $\alpha \in S$ we set
  \[s(T,V,\vectorletter{t}) = \alpha (a_1(T,V,\vectorletter{t}),\ldots,a_k(T,V,\vectorletter{t})) \]
  and for $s = m$ with $1\leq m\leq n$  we define
  $s(T,V,\vectorletter{t}) = \vectorletter{t}(m)$.
  We call a natural transformation of the form $s\in S(n)$ \emph{canonical}.
\end{definition}
The natural transformations of \autoref{ex:numbers_transformation} yield examples of  
canonical transformations.
We now define representations of the 1--signature $(S_{\TLC},\Sigma_{\TLC})$ of the simply--typed lambda calculus. 
Afterwards we define general 1--signatures and their representations.

\begin{example}[\autoref{ex:tlc_mod_higher_degree} cont.]\label{ex:tlc_mod_mor_higher_degree}
 Let $S := S_{\TLC}$ be the signature for types of $\TLC$ as in \autoref{ex:type_sig_SLC}.
 We denote by $i : 1 \Rightarrow SU_2 : S\C_2 \to \Set$, for $i = 1,2$, the natural transformations
 defined analogously to those of \autoref{ex:tlc_mod_higher_degree}.
 We define the transformation $1 \TLCar 2 : 1 \Rightarrow SU_2$ as
   \[ (1 \TLCar 2) (V,s,t)(*) := s \TLCar t \enspace . \]
 The constructors of the simply--typed lambda calculus thus constitute the carriers of 
two module morphisms,
\begin{align}
 \Abs &: \fibre{{\TLCDELTA_2}^1}{2} \to \fibre{\TLCDELTA_2}{1 \TLCar 2} \notag \\
 \App &: \fibre{{\TLCDELTA_2}}{1 \TLCar 2} \times \fibre{\TLCDELTA_2}{1} \to \fibre{\TLCDELTA_2}{2} \enspace . \label{eq:tlcd_constr}
\end{align}

   \label{ex:rep_1--sig_tlc}

\noindent
 Altogether we model the simply--typed lambda calculus with equality relation
 via the following categorical structure:
 \begin{itemize}
  \item the relative monad $\TLCDELTA$ on $\TDelta{\TLCTYPE}$ and
  \item two morphisms of $\TLCDELTA_2$--modules $\Abs$ and $\App$ of type as in \autoref{eq:tlcd_constr}.
 \end{itemize}
 We thus define a \emph{representation} of the simply--typed lambda calculus, specified by 
the signature $(S_{\TLC},\Sigma_{\TLC})$ (cf.\ \autoref{eq:sig_tlc_higher_order}), as
     a representation $T$ of $S_{\TLC}$ in a set $T$, a monad $P$ on $\TDelta{T}$ and 
     two morphisms of $P_2$--modules 
  \begin{align*}\Abs : \fibre{{P_2}^1}{2} \to \fibre{P_2}{1\TLCar 2}  \quad\text{ and } \quad
            \App :\fibre{{P_2}}{1 \TLCar 2} \times \fibre{P_2}{1} \to \fibre{P_2}{2} \enspace . 
  \end{align*}
              
\noindent
Together with a suitable definition of \emph{morphisms of representations}, this yields a category
 in which the triple $(\TLCDELTA, \Abs,\App)$ is the initial object.
\end{example}

In general, an arity over $S$ of degree $n\in \NN$ is given by a pair of functors, each of which
associates, to any suitable monad $P$, a source $\dom(s,P)$ and a target $\dom(s,P)$ 
of a $P_n$--module morphism.
Each such functor is called a \emph{half--arity}.
Representing an arity in the monad $P$ then means specifying a module morphism
$\dom(s,P)\to\cod(s,P)$.

We define the source and target categories of half--arities;
an object of the source category is a pair of a representation of $S$ 
in a set $T$ and a monad on $\TDelta{T}$.

\begin{definition}[Relative $S$--Monad] \label{def:s-rmon}
  The \emph{category $\SigRMon{S}$ of relative $S$--monads} is the category whose objects are pairs $(T,P)$ of
  a representation $T$ of $S$ and a relative monad $P$ on $\TDelta{T}$.
  A morphism from $(T,P)$ to $(T', P')$ is a pair $(g, f)$ of a morphism of $S$--re\-pre\-sen\-ta\-tions $g : T\to T'$ and a 
    morphism of relative monads $f : P\to P'$ over $\retyping{g}$ as in \autoref{rem:rel_mon_mor_case}.

   Given $n\in \mathbb{N}$, we write $\SigRMon{S}_n$ for the category whose objects are pairs $(T,P)$ of a representation $T$ of $S$ and 
  a relative monad $P$ over $\TDelta{T}_n$. A morphism from $(T,P)$ to $(T', P')$ is a pair $(g, f)$ of a morphism of 
      $S$--representations $g : T\to T'$ and a 
   monad morphism $f : P\to P'$ over the retyping functor $\retyping{g}(n)$ (\autoref{def:retyping_functor_pointed}).
\end{definition}

The target categories mix modules over different relative monads:

\begin{definition}
  \label{def:lrmod_typed}
  Given $n\in \mathbb{N}$, an algebraic signature $S$ and a category $\D$, 
  we call $\LRMod{n}{S}{\D}$ 
  the category an object of which is a pair $(P,M)$ of a relative $S$--monad $P \in \SigRMon{S}_n$ and a $P$--module with codomain $\D$.
  A morphism to another such $(Q,N)$ is a pair $(f, h)$ of a morphism of relative $S$--monads 
$f : P \to Q$ in $\SigRMon{S}_n$ and a morphism of relative modules $h : M \to f^*N$.
\end{definition}

We sometimes just write the module --- i.e.\ the second --- component of an object or morphism
 of the large category of modules.
Given $M\in \LRMod{n}{S}{\D}$, we thus write $M(V,\vectorletter{t})$ or $M_{V,\vectorletter{t}}$ 
for the value of the module on the object $(V,\vectorletter{t})$.

A \emph{half--arity over $S$ of degree $n$} is a functor from relative $S$--monads to the category of large modules of degree $n$.

\begin{definition}
  \label{def:half_arity_degree_semantic_typed}
 Given an algebraic signature $S$ and $n\in \mathbb{N}$, a \emph{half--arity over $S$ of degree $n$} is a functor
  $\alpha : \SigRMon{S} \to \LRMod{n}{S}{\PO}$
 which is pre--inverse to the forgetful functor.
\end{definition}

\noindent
The basic building brick for the half--arities we consider is the \emph{tautological} module:

\begin{definition}
\label{def:taut_mod_pointed}
To any relative $S$--monad $R$ we associate
  the \emph{tautological module} of $R_n$ (cf.\ \autoref{def:cat_indexed_pointed}), 
  \[\Theta_n(R):= (R_n,R_n) \in \LRMod{n}{S}{\TP{T}_n} \enspace . \]
\end{definition}

\noindent
From the tautological module, we build \emph{classic} half--arities using 
canonical natural transformations (cf.\ \autoref{def:canonical_nat_trans}); 
these transformations specify context extension (derivation) and
selection of specific object types (fibre):

\begin{definition}[Classic Half--Arity]
  The following clauses define an inductive set of 
  \emph{classic} half--arities, to which we restrict our attention:
  \begin{itemize}
   \item The constant functor $* : R \mapsto 1_{\RMod{R}{\PO}}$ is a classic half--arity.
   \item Given any canonical natural transformation $\tau : 1 \Rightarrow S U_n$ (cf.\ \autoref{def:canonical_nat_trans}), 
         the point-wise fibre module with respect to $\tau$ (cf.\ \autoref{def:fibre_rel_mod_II}) of the tautological module 
         $\Theta_n : R\mapsto (R_n, R_n)$ (cf.\ \autoref{def:taut_mod_pointed}) is a classic half--arity of degree $n$, 
        \[ \fibre{\Theta_n}{\tau} : \SigRMon{S} \to \LRMod{n}{S}{\PO} \enspace , \quad R\mapsto \fibre{R_n}{\tau} \enspace . \]
   \item Given any (classic) half--arity $M : \SigMon{S} \to \LMod{n}{S}{\PO}$ 
         of degree $n$ and a canonical natural transformation $\tau : 1 \Rightarrow S U_n$, 
         the point-wise derivation of $M$ with respect to $\tau$ is a (classic) half--arity of degree $n$, 
      \[ M^{\tau} : \SigRMon{S} \to \LRMod{n}{S}{\PO} \enspace , \quad R\mapsto \bigl(M(R)\bigr)^{\tau} \enspace . \]
   \item 
      Given two (classic) half--arities $M$ and $N$ of degree $n$, 
       their pointwise product of modules $M\times N$ is again a (classic) half--arity of degree $n$. 
  \end{itemize}

\end{definition}

\noindent
A half--arity of degree $n$ thus associates, to any relative $S$--monad $P$ over a set of types $T$, 
a \emph{family of $P$--modules} indexed by $T^n$, cf.\ \autoref{rem:family_of_mods_cong_pointed_mod_relative}.

An arity of degree $n\in \mathbb{N}$ for terms over an algebraic signature $S$ is defined to be a pair of functors
from relative $S$--monads to modules in $\LRMod{n}{S}{\PO}$.
The degree $n$ corresponds to the number of object type indices of its associated constructor.
For instance, the arities of $\Abs$ and $\App$ of \autoref{eq:sig_tlc_higher_order} are of degree $2$.

\begin{definition}[Term--Arity, Signature over $S$]
  A \emph{classic arity $\alpha$ over $S$ of degree $n$} is a pair 
 $ s = \bigl(\dom(\alpha), \cod(\alpha)\bigr) $
  of half--arities over $S$ of degree $n$ such that 
$\dom(\alpha)$ is classic and
$\cod(\alpha)$ is of the form $\fibre{\Theta_n}{\tau}$ for some canonical 
        transformation $\tau$ as in \autoref{def:canonical_nat_trans}.
We write $\dom(\alpha) \to \cod(\alpha)$ for the arity $\alpha$, and $\dom(\alpha, R) := \dom(\alpha)(R)$
and similar for the codomain and morphisms of relative $S$--monads. 
A \emph{term--signature} $\Sigma$ over $S$ 
is a family of classic arities (of varying degree) over $S$. 
\end{definition}

\begin{definition}[1--Signature]\label{def:1--sig}
  A \emph{1--signature} is a pair $(S,\Sigma)$ consisting of an algebraic signature $S$ for sorts and 
  a term--signature $\Sigma$ over $S$.
 \end{definition}

\begin{example}[\autoref{ex:tlc_mod_mor_higher_degree} cont.]\label{ex:tlc_sig_higher_order}
The terms of the simply typed lambda calculus over the type signature of \autoref{ex:type_sig_SLC} are given by the arities
 \begin{equation*}     
   \abs : \fibre{\Theta}{2}^1 \to \fibre{\Theta}{1\TLCar 2} \enspace , \quad
   \app : \fibre{\Theta}{1\TLCar 2} \times \fibre{\Theta}{1} \to \fibre{\Theta}{2} \quad ,
 \end{equation*}
  both of which are of degree $2$ --- we leave the degree implicit. 
  The outer lower index and the exponent are to be interpreted as de Bruijn variables, ranging over types. 
  They indicate the fibre (cf.\ \autoref{def:fibre_rel_mod_II}) and derivation (cf.\ \autoref{def:derived_rel_mod_II}),
  respectively, in the special case where the corresponding natural transformation is given by a natural number
  as in \autoref{def:nat_trans_type_indicator}.
\end{example}

\begin{example}[\autoref{ex:type_PCF} cont.] \label{ex:term_sig_pcf}
  The term--signature of $\PCF$ consists of an arity for abstraction and an arity for application,
each of degree 2, an arity (of degree 1) for the fixed point operator, and 
one arity of degree 0 for each logic and arithmetic constant --- some of which we omit:
\begin{align*}
     \abs &: \fibre{\Theta}{2}^1 \to \fibre{\Theta}{1\PCFar 2} \enspace , \quad
       \app : \fibre{\Theta}{1\PCFar 2} \times \fibre{\Theta}{1} \to \fibre{\Theta}{2} \enspace ,\quad 
       \PCFFix : \fibre{\Theta}{1\PCFar 1} \to \fibre{\Theta}{1} \enspace ,  \\
      &\PCFn{n} : * \to \fibre{\Theta}{\Nat }  \quad \text{for $n\in \NN$} \enspace , \quad
        \PCFSucc : * \to \fibre{\Theta}{\Nat \PCFar \Nat} \enspace , \quad
        \PCFZerotest : * \to \fibre{\Theta}{\Nat \PCFar \Bool } 
\end{align*}
\end{example}

\begin{definition}[Representation of an Arity, of a 1--Signature over $S$]
   \label{def:1--rep_typed}
   A representation $r$ of an arity $\alpha$ over $S$ in an $S$--monad $R$ is a morphism of relative modules 
     $r : \dom(\alpha,R) \to \cod(\alpha, R)$.
  A representation $R$ of a signature over $S$ is a given by a relative $S$--monad --- called $R$ as well ---  
  and a representation $\alpha^R$ of each arity $\alpha$ of $S$ in $R$.
\end{definition}

Representations of $(S,\Sigma)$ are the objects of a category $\Rep^\Delta(S,\Sigma)$, whose morphisms are defined as follows:
\begin{definition}[Morphism of Representations]
  \label{def:rel_mor_of_reps_typed}
  Given representations $P$ and $R$ of a typed signature $(S,\Sigma)$, a morphism of representations 
  $f : P\to R$ is given by a morphism of relative $S$--monads $f : P \to R$, such that for any arity $\alpha$ of $\Sigma$
  the following diagram of module morphisms commutes:
  \[
  \comp{\alpha^P}{\cod(\alpha,f)} = \comp{\dom(\alpha,f)}{\alpha^R} \enspace . 
%
 \]
\end{definition}

\begin{lemma}\label{lem:init_no_eqs_typed} 
  For any 1--signature $(S,\Sigma)$, the category of representations of $(S,\Sigma)$ has an initial object.
\end{lemma}

\begin{proof} 
  The initial object is obtained, analogously to the untyped case (cf.\ \cite{ahrens_relmonads}), 
 via an adjunction $\Delta_* \dashv U_*$ between the categories of representations of $(S,\Sigma)$
 in relative monads and those in monads as in \cite{ahrens_ext_init}. We use that left adjoints are
 cocontinuous, and thus preserve initial objects.
\end{proof}

\subsubsection{Inequations \& 2--Signatures}

An inequation associates, 
to any representation of $(S,\Sigma)$ in a relative monad $P$, two parallel morphisms of $P$--modules.
Similarly to arities, an inequation (of higher degree) may be given by a \emph{family of inequations}, indexed by object types.
Consider the simply--typed lambda calculus, which was defined with \emph{typed} abstraction and application.
Similarly, we have a \emph{typed substitution} operation for $\TLC$ and,
more generally, for any monad on $\TDelta{T}$ (cf.\ \autoref{def:hat_P_subst_typed}). 
For $s,t\in \TLCTYPE$ and $M\in\SLC(V^{*s})_t$ and $N \in \SLC(V)_s$, beta reduction is specified by

  \[ \lambda M(N) \leadsto M [* := N] \enspace , \]
where our notation hides the fact that abstraction, application and substitution are typed operations.
More formally, such a reduction rule might read as a family of inequations between morphisms of modules
\[  \comp{(\abs_{s,t} \times\id)}{\app_{s,t}} \enspace \leq \enspace \_ [*^s :=_t \_ ] \enspace , \] 
where $s,t\in \TLCTYPE$ range over types of the simply--typed lambda calculus.
Analogously to 1--signatures, we want to specify the beta rule without referring to the set $\TLCTYPE$,
but instead express it for an arbitrary representation $R$ of the typed signature $(S_{\TLC},\Sigma_{\TLC})$ 
(cf.\ \autoref{ex:tlc_sig_higher_order}),
as in

\[  \comp{(\abs^R \times \id)}{\app^R} \enspace \leq \enspace \_ [* := \_ ] \enspace , \] 
where both the left and the right side of the inequation are given by suitable $R$--module morphisms of degree 2.

\begin{definition}
 Let $(S,\Sigma)$ be a 1--signature, and let $U : \Rep^{\Delta}(S,\Sigma)\to \SigRMon{S}$ be the forgetful functor.
 Given two (classic) half--arities $\dom(s)$ and $\cod(s)$ of degree $n\in \mathbb{N}$, 
a \emph{half--equation} $\alpha : \dom(s) \to \cod(s)$ of degree $n$ over $(S,\Sigma)$
is a natural transformation $\alpha : \comp{U}{\dom(s)} \to \comp{U}{\cod(s)}$.
We call an inequation \emph{classic} when its codomain is given by a classic half--arity.
\end{definition}

\begin{definition}[Substitution of \emph{one} Variable as a Half--Equation]
\label{def:forget_hat_module}
   \label{def:hat_P_subst_typed}
  Let $T$ be a (nonempty) set and let $P$ be a monad over $\family{\Delta}{T}$.
  For any $s,t\in T$ and $X\in \TS{T}$ we define a binary substitution operation
    $(y,z)\mapsto y [*:= z] := \kl{[\we_X , x \mapsto z] }(y)$.
  For any pair $(s,t)\in T^2$, we thus obtain a morphism of $P$--modules
  \[ \subst^{P}_{s,t} : \fibre{{P}^{s}}{t} \times \fibre{{P}}{s} \to \fibre{{P}}{t} \enspace . \]
\label{def:subst_half_eq_typed}
  
\noindent
By \autoref{rem:family_of_mods_cong_pointed_mod_relative} this family is equivalent to a module morphism of degree 2.
 We thus have a half--equation of degree $2$ with classic domain and codomain over any typed signature,
 \[\subst  : R\mapsto \subst^R :  \fibre{{R}_2^1}{2} \times \fibre{{R}_2}{1} \to \fibre{{R}_2}{2} \enspace . \]
\end{definition} 

\begin{example}[\autoref{ex:tlc_sig_higher_order} cont.]\label{ex:app_circ_half_typed}
  The following map yields a half--equation over the signature $\SLC$, as well as over the signature of $\PCF$:
 \[ \comp{(\abs\times\id)}{\app} : R \mapsto \comp{(\abs^R \times \id^R)}{\app^R}  : 
\fibre{{R}_2^1}{2} \times \fibre{{R}_2}{1} \to \fibre{{R}_2}{2} \enspace . \]
\end{example}

\begin{definition}[Inequation] \label{def:ineq_typed}
 Given a signature $(S,\Sigma)$, an \emph{inequation over $(S,\Sigma)$}, or \emph{$(S,\Sigma)$--inequation}, 
of degree $n\in \NN$ is a pair of 
parallel half--equations of degree $n$.
We write $\alpha \leq \gamma$ for the inequation $(\alpha, \gamma)$.
\end{definition}

\begin{example}[Beta Reduction]
  For any suitable 1--signature --- i.e.\ for any 1--signature that has an arity for abstraction and an arity for application ---
  we specify beta reduction using the parallel half--equations of \autoref{def:subst_half_eq_typed} and \autoref{ex:app_circ_half_typed}:
  \[   \comp{(\abs \times \id)}{\app} \leq \subst : \fibre{\Theta_2}{2}^1 \times \fibre{\Theta_2}{1} \to \fibre{\Theta_2}{2} \enspace . \]
\end{example}

\begin{example}[Fixpoints and Arithmetics of $\PCF$]\label{ex:pcf_ineqs}
 We specify some of the reduction rules of $\PCF$ via inequations over the 1--signature of $\PCF$ (cf.\ \autoref{ex:term_sig_pcf});
for space reasons we refrain from specifying all of them. 
 The reader may fill in the missing inequations, whose informal specification can be found, e.g., in \cite{Hyland00onfull}.
   \begin{align*}
      \PCFFix \enspace &\leq \enspace \comp{(\id,\PCFFix)}{\app} \enspace : \enspace \fibre{\Theta}{1\PCFar 1} \to \fibre{\Theta}{1} \\
      \comp{(\PCFPred,\PCFn{0})}{\app} \enspace &\leq \enspace \PCFn{0} \enspace : \enspace * \to \fibre{\Theta}{\Nat} \\
       \quad \comp{\left(\PCFPred,\comp{({\PCFSucc},{\PCFn{n}})}{\app}\right)}{\app} \enspace  &\leq \enspace  \PCFn{n} \enspace : \enspace * \to \fibre{\Theta}{\Nat}\\
      \comp{(\PCFZerotest,\PCFn{0})}{\app} \enspace &\leq \enspace  \PCFTrue \enspace : \enspace * \to \fibre{\Theta}{\Bool} \\
          \quad \comp{\left(\PCFZerotest,\comp{(\PCFSucc,\PCFn{n})}{\app}\right)}{\app} \enspace  &\leq \enspace \PCFFalse \enspace : \enspace * \to \fibre{\Theta}{\Bool}
   \end{align*}

\end{example}

\begin{definition}[Representation of Inequations]\label{def:rep_ineq_typed}
  \label{def:2--rep_typed}
 A \emph{representation of an $(S,\Sigma)$--inequa\-tion $\alpha\leq \gamma : U \to V$} (of degree $n$) is any representation 
  $R$ over a set of types $T$ of $(S,\Sigma)$ such that 
  $\alpha^R \leq \gamma^R$ pointwise, i.e.\ if for any pointed context $(X,\vectorletter{t}) \in \TS{T}\times T^n$, 
      any $t\in T$ and any $y\in U^R_{(X, \vectorletter{t})}(t)$, 
 $\alpha^R(y) \enspace \leq \enspace \gamma^R(y)$, 
where we omit the sort argument $t$ as well as the context $(X,\vectorletter{t})$ from $\alpha$ and $\gamma$.
We say that such a representation $R$ \emph{satisfies} the inequation $\alpha \leq \gamma$.

The category of representations of $((S,\Sigma), A)$ is defined as the full subcategory of $\Rep^{\Delta}(S,\Sigma)$ of 
representations  satisfying each inequation of $A$.
According to \autoref{rem:family_of_mods_cong_pointed_mod_relative}, 
 the above inequation is equivalent to ask whether, for any $\vectorletter{t} \in T^n$, 
    any $t\in T$ and any $y\in U_{\vectorletter{t}}^R(X)(t)$, 
$\alpha_{\vectorletter{t}}^R(y) \enspace \leq \enspace \gamma_{\vectorletter{t}}^R(y)$.
\end{definition}

\begin{definition}[2--Signature]\label{def:2--sig}
A \emph{2--signature} is a pair given by a 1--signature $(S,\Sigma)$ and a set $A$ of \emph{classic}
 inequations over $(S,\Sigma)$.
\end{definition}

\begin{example}[Representations of $\TLC$ with $\beta$]
 \label{ex:rep_2--sig_tlc}
 A representation of $(S_{\TLC},\Sigma_{\TLC},\beta)$ is given by a representation
  $(P, \Abs, \App)$ of $(S_{\TLC},\Sigma_{\TLC})$  
  over a set $T$ ``of types'' such that,
  for any context $V\in \TS{T}$, any $s,t\in T$ and any $M\in P^s(V)(t)$ and $N\in P(V)(s)$,
   \[ \App_{s,t}(\Abs_{s,t}(M),N) \enspace \leq \enspace M[*:= N] \enspace . \]
 The initial such representation is given by the triple $(\TLCB, \Abs, \App)$, where
 \begin{align*}
    \Abs &: \fibre{{\TLCB_2}^1}{2} \to \fibre{\TLCB_2}{1 \TLCar 2} \\
 \App &: \fibre{{\TLCB_2}}{1 \TLCar 2} \times \fibre{\TLCB_2}{1} \to \fibre{\TLCB_2}{2} \enspace .
 \end{align*}

\end{example}

\noindent
The above example is an instance of the following general theorem for 2--signatures:

\begin{theorem}\label{thm:init_w_ineq_typed}
 For any set of classic $(S,\Sigma)$--inequations $A$, the category of representations of $((S,\Sigma),A)$ has an initial object.
\end{theorem}
A proof of the theorem can be found in the author's PhD thesis \cite{ahrens_phd}.

The following remark gives a ``manual''
on how to use the universal property of initiality in order to specify a translation 
between two languages:

\begin{remark}[Iteration Principle by Initiality]\label{rem:comp_sem_iteration}
 The universal property of the language generated by a 2--signature yields an \emph{iteration principle}
  to define maps --- translations --- on this language, which are compatible by construction with substitution and 
  reduction in the source and target languages. 
A translation from the language generated by $(S,\Sigma,A)$ to the language 
generated by $(S',\Sigma',A')$ can be obtained, via the universal property, 
as an initial morphism in $\Rep^{\Delta}(S,\Sigma,A)$, 
obtained by equipping the relative monad $\init{\Sigma}'_{A'}$ underlying the target language with a representation
of the signature $(S,\Sigma,A)$. In more detail:
\begin{enumerate}
 \item we give a representation of the type signature $S$ in the set $\init{S}'$. By initiality of $\init{S}$, this yields a 
       translation $\init{S} \to \init{S}'$ of sorts.
 \item Afterwards, we specify a representation of the term signature $\Sigma$ in the monad $\init{\Sigma}'_{A'}$ by 
        defining suitable (families) of morphisms of $\init{\Sigma}'_{A'}$--modules. This yields a representation $R$
           of $(S, \Sigma)$ in the monad $\init{\Sigma}'_{A'}$.
 \item Finally, we verify that the representation $R$ of $(S,\Sigma)$ satisfies the inequations of $A$, that is, we check
       whether, for each $\alpha \leq \gamma : \dom(\alpha) \to \cod(\alpha) \in A$, and for each context $V$, each $t\in \init{S}$ and $x \in \dom(\alpha)^R_V(t)$,
        $ \alpha^R (x) \enspace \leq \enspace \gamma^R (x)$.
\end{enumerate}
\end{remark}

\begin{example}[Translation from $\PCF$ to $\LC$]
We use the aforementioned iteration principle to specify a translation from $\PCF$ to $\ULC$, which is semantically faithful with respect to the usual 
reduction relation of $\PCF$ --- generated by the inequations of \autoref{ex:pcf_ineqs} (and some more, see \cite{Plotkin1977223}) ---
and beta reduction of $\ULC$.
For space reasons, we cannot present this example here; 
we refer to \cite{ahrens_phd}.
This example --- initiality of the types and terms of $\PCF$ with its reductions, and a translation to $\ULC$ with beta reduction via associated category--theoretic iteration operator ---
has also been implemented in the proof assistant \textsf{Coq}. 
The source files and documentation are available on \sourceurl.
\end{example}

\appendix

\bibliographystyle{splncs03}
\bibliography{literature}
\end{document}